\documentclass[conference]{IEEEtran}
\newcommand{\figurewidth}{0.8\linewidth}

\usepackage{cite}
\usepackage{graphicx}
\usepackage{epsfig,psfrag}
\usepackage{balance}
\usepackage{array}
\usepackage{theorem}
\usepackage{amsmath}
\usepackage{amssymb}
\usepackage{flushend}
\usepackage{tikz}

\newtheorem{theorem}{Theorem}

\newtheorem{lemma}{Lemma}
\newtheorem{proposition}[theorem]{Proposition}

\newtheorem{definition}{Definition}

\newcommand{\mc}{\mathcal}
\newcommand{\mb}{\mathbf}

\newcommand{\NM}{\text{UComp}}
\newcommand{\CM}{\text{UCompM}}
\newcommand{\DCM}{\text{DUCompM}}

\newcommand{\WDCM}{\text{DUCompM}}
\newcommand{\MM}{m}
\newcommand{\SW}{Slepian-Wolf}

\newcommand{\PE}{p_e}

\begin{document}
\IEEEoverridecommandlockouts
%
\title{On Lossless Universal Compression of\\Distributed Identical Sources\vspace{-0.06in}}
\author{Ahmad Beirami and Faramarz Fekri\\School of Electrical and Computer Engineering\\Georgia Institute of Technology,~Atlanta~GA~30332, USA\\Email:~\{beirami,~fekri\}@ece.gatech.edu

\thanks{This material is based upon work supported by the National Science Foundation under Grant No. CNS-1017234.}}
%
%
%

\maketitle
\thispagestyle{empty}
\pagestyle{empty}

\begin{abstract}
Slepian-Wolf theorem is a well-known framework that targets almost lossless compression of (two) data streams with symbol-by-symbol correlation between the outputs of (two) distributed sources. However, this paper considers a different scenario which does not fit in the Slepian-Wolf framework. We consider two identical but spatially separated sources. We wish to study the universal compression of a sequence of length $n$ from one of the sources provided that the decoder has access to (i.e., memorized) a sequence of length $m$ from the other source. Such a scenario occurs, for example, in the universal compression of data from multiple mirrors of the same server. In this setup,  the correlation does not arise from symbol-by-symbol dependency of two outputs from the two sources. Instead, the sequences are correlated through the information that they contain about the unknown source parameter. We show that the finite-length nature of the compression problem at hand requires considering a notion of almost lossless source coding, where coding incurs an error probability $\PE(n)$ that vanishes with sequence length $n$. We obtain a lower bound on the average minimax redundancy of almost lossless codes  as a function of the sequence length $n$ and the permissible error probability $\PE$ when the decoder has a memory of length $\MM$ and the encoders do not communicate. Our results demonstrate that a strict performance loss is incurred when the two encoders do not communicate even when
the decoder knows the unknown parameter vector (i.e., $\MM \to \infty$).
\end{abstract}


\section{Introduction}
\label{sec:intro}
Many practical applications involve compression of data that are taken from multiple spatially separated sources. A key challenge in most of such applications is that the sources usually cannot communicate with each other. Theoretical results by Slepian and Wolf demonstrate that if the data streams from two sources have symbol-by-symbol correlation, the sequences can be compressed to their joint entropy even when the two encoders do not communicate~\cite{slepian_wolf}. In other words, as in Fig.~\ref{fig:distributed_basic}, assume that sources $S_1$ and $S_2$ wish to transmit the sequences $y^n$ and $x^n$, respectively, to a node $R$. As the length $n$ of the sequences increases, the decoding of $x^n$ at $R$ with the help of $y^n$ can be performed using a code with the average length that asymptotically approaches the conditional entropy, (i.e., $H(X^n| Y^n)$) with asymptotically zero error probability.
If the decoder did not choose to use $y^n$ in decoding, the encoder at $S_2$ would  have to encode the sequence $x^n$ irrespective to $y^n$ with an average length that is lower bounded by $H(X^n)$.
Note that the conditional entropy $H(X^n| Y^n)$ may be significantly smaller than the individual entropy $H(X^n)$. After recent development of practical {\SW} (SW) coding schemes by Pradhan and Ramchandran~\cite{PRADHAN-DSC}, SW coding has drawn a great deal of attention as a promising technique for sensor networks~\cite{Mina_TCOM} and distributed video coding~\cite{Distributed_Video_Coding}.

The {\SW} theorem naturally suits applications where the (new) sequence $x^n$ from $S_2$ (in Fig.~\ref{fig:distributed_basic}) can be viewed as a noisy version of the (previously seen) sequence $y^\MM$ that could possibly be exploited as side information to reduce the code length of $x^n$. Data gathering from sensors that measure the same phenomenon is one example.
However, in many scenarios,
the compression of distributed sources cannot be modeled by the SW framework.
As an example, consider the universal compression of data from the mirrors of the same server, where the sources are exact copies of each other. Hence, it is plausible to assume that the sources ($S_1$ and $S_2$ in Fig.~\ref{fig:distributed_basic}) follow the same statistical model. On the other hand, the source model might be unknown requiring universal compression~\cite{Universal_Finite_memory_source, Rissanen_1984,ISIT11}.
The question is, assuming two identical sources $S_1$ and $S_2$ and having $y^m$ from $S_1$ at the decoder, what is the achievable universal compression performance on $x^n$ at $S_2$ provided that the encoders at $S_1$ and $S_2$ do not communicate.

\begin{figure}
\centering
\vspace{0.05in}
\includegraphics[height=1.2in, angle=-90]{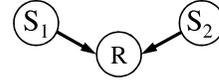}
\caption{The basic scenario for the compression of distributed sources.}
\vspace{0.05in}
\label{fig:distributed_basic}
\vspace{-0.3in}
\end{figure}

We stress that the nature of this problems is fundamentally different from those addressed by the {\SW} (SW) theorem in~\cite{slepian_wolf}.
 Here,  instead of symbol-by-symbol correlation between the sequences as in SW setup,
 the redundancy is due to the fact that when the source parameter is a priori unknown there is significant overhead in the universal compression of finite-length sequences~\cite{ISIT11, Merhav_Feder_IT, DCC12}.
Considering the example in Fig.~\ref{fig:distributed_basic} with two identical sources $S_1$ and $S_2$, $y^m$ and $x^n$ would be independent given that the source model is known. However, when the source parameter is unknown, $y^m$ and $x^n$ are {\em correlated} with each other through the information they contain about the unknown source parameter.
 The question is whether or not this correlation can be potentially leveraged by the encoder of $S_2$ and the decoder at  $R$ in the decoding of $x^n$ using $y^\MM$ in order to reduce the code length of $x^n$.

In this paper, we study the universal compression of distributed identical sources. By identical we mean that the sources ($S_1$ and $S_2$) share the same unknown source parameter. By distributed we mean that the sources are spatially separated and the encoders do not communicate with each other.
This problem can also be viewed as universal compression with training data that is only available to the decoder. It is known that forming a statistical model from a training data set would improve the performance of universal compression~\cite{Korodi2005,INFOCOM12}. In~\cite{DCC12, ISIT12_gain}, we theoretically derived the \emph{gain} that is obtained in the universal compression of the new sequence $x^n$ from $S_2$ by memorizing (i.e., having access to) $y^m$ from $S_1$ at \emph{both} the decoder (at $R$) and the encoder (at $S_2$). This corresponds to the reduced case of our problem where the sources $S_1$ and $S_2$ are either co-located (a single source) or allowed to communicate. For the reduced problem case, in~\cite{ITW11, INFOCOM12}, we further extended the setup to a network with a single source and derived bounds on the \emph{network-wide gain} where a small fraction of the intermediate nodes in the network are capable of memorization.
However, as we demonstrate in the present paper, the extension to the multiple spatially separated sources, where the training data is only available to the decoder, is non-trivial and raises a new set of challenges  that we aim to address.
The rest of this paper is organized as follows. In Sec.~\ref{sec:background}, we briefly review the necessary background. In Sec.~\ref{sec:setup}, we describe the problem setup. In Sec.~\ref{sec:results}, we present our main results. In Sec.~\ref{sec:discussion}, we provide discussion on the results. In Sec.~\ref{sec:analysis}, we present the technical analysis of the results. Finally Sec.~\ref{sec:conclusion} concludes the paper.

\section{Background Review}
\label{sec:background}
In this section, we review the necessary background, notations, and definitions followed by the formal problem setup.
Following the notation in~\cite{ISIT12_gain}, let $\mc{A}$ be a finite alphabet. Let $d$ be the number of the source parameters. Further, let $\theta = (\theta_1,...,\theta_d)$ denote the $d$-dimensional parameter vector associated with the parametric source (that is a priori unknown). Let $\Theta^d$ denote the space of $d$-dimensional parameter vectors. We denote $\mu_\theta$ as the probability measure that is defined by the parameter vector $\theta$.
Let $\mc{P}^d$ denote the \emph{family} of sources that are described with the $d$-dimensional unknown parameter vector $\theta \in \Theta^d$. We use the notation $x^n = (x_1,...,x_n) \in \mc{A}^n$ to present a sequence of length $n$ from the alphabet $\mc{A}$. We further denote $X^n$ as a random sequence of length $n$ (that follows the probability distribution $\mu_\theta$).
 Let $H_n(\theta)$ be the source entropy given $\theta$, i.e., $H_n(\theta) =  \mb{E} \log \left(\frac{1}{\mu_\theta(X^n)} \right)$.\footnote{Throughout this paper, all expectations are taken with respect to the probability measure $\mu_\theta$, and $\log(\cdot)$  denotes the logarithm in base $2$.}

Let $c_n:\mc{A}^n \to \{0,1\}^*$ be an injective mapping from the set $\mc{A}^n$ of the sequences of length $n$ over $\mc{A}$ to the set $\{0,1\}^*$ of binary sequences.
Next, we present the notions of strictly lossless and almost lossless source codes, which will be needed for the study of UC-DIS.
\begin{definition}
The code $c_n(\cdot): \mc{A}^n \to \{0,1\}^*$ is called strictly lossless (also called zero-error) if there exists a reverse mapping $d_n(\cdot):\{0,1\}^* \to \mc{A}^n$ such that
$$\forall x^n \in \mc{A}^n:~~d_n(c_n(x^n)) = x^n.$$
\end{definition}
All of the practical data compression schemes are examples of strictly lossless codes, namely, the arithmetic coding, Huffman, Lempel-Ziv, and Context-Tree-Weighting algorithms.

On the other hand, due to the distributed nature of the sources, we are concerned with the slightly weaker notion of almost lossless source coding in this paper.
\begin{definition}
The code $\hat{c}_n^{\PE}(\cdot): \mc{A}^n \to \{0,1\}^*$ is called almost lossless with permissible error probability $\PE(n)=o(1)$, if there exists a reverse mapping $\hat{d}_n^{\PE}(\cdot):\{0,1\}^* \to \mc{A}^n$ such that
$$\mb{E}\{ \mb{1}_e(X^n) \}\leq \PE(n),$$
where $\mb{1}_e(x^n)$ denotes the error indicator function, i.e,
\begin{equation}
\mb{1}_e(x^n) =
\left\{
\begin{array}{ll}
1 & \hspace{0.1in} \hat{d}_n^{\PE}(\hat{c}_n^{\PE}(x^n)) \neq x^n,\\
0 & \hspace{0.1in} \text{otherwise}.
\end{array}
\right.\nonumber
\end{equation}
\end{definition}
 The almost lossless codes allow a non-zero error probability $\PE(n)$ for any $n$ while they are \emph{almost surely} asymptotically error free. Note that strictly lossless codes correspond to $\PE(n) = 0$. 
The proofs of Shannon~\cite{Shannon_paper} for the existence of entropy achieving source codes are based on almost lossless random codes. Further, the proof of the SW theorem~\cite{slepian_wolf} also uses almost lossless codes. Further, all of the practical implementations of SW source coding are based on almost lossless codes (cf.~\cite{PRADHAN-DSC,Mina_TCOM}).
We stress that the nature of the almost lossless source coding is different from that incurred by the lossy source coding (i.e., the rate-distortion theory). In the rate-distortion theory, a code is designed to  asymptotically achieve a given distortion level  as the length of the sequence grows to infinity. Therefore, since the almost lossless coding asymptotically achieves a zero-distortion, in fact, it coincides with the special case of zero-distortion in the rate-distortion curve.

\vspace{-.1in}
\section{Problem Setup}
\label{sec:setup}
We present the problem setup in the most basic scenario,  shown in Fig.~\ref{fig:distributed_basic},  consisting of two identical sources located in nodes $S_1$ and $S_2$, and the destination node $R$. We let the information sources at $S_1$ and $S_2$ be parametric with an identical $d$-dimensional parameter vector that is unknown a priori to the encoder and the decoder. Let $y^\MM$ and $x^n$ denote two sequences with lengths $\MM$ and $n$, respectively, that are generated by the unknown information source model.
In the sequel, we describe the communication scenario for universal compression of distributed identical sources. We assume that $S_1$ has transmitted the sequence $y^\MM$ to $R$. Next, at some later time, $S_2$ wishes to send $x^n$ to $R$. We further assume that $R$ is a memory unit and is capable of memorizing the sequence $y^\MM$.
We investigate the achievable saving in the compression of $x^n$  in the $S_2$-$R$ link when $R$ has memorized the sequence $y^\MM$. Note that $S_2$ does not have access to the sequence $y^\MM$.
 If the node $R$ did not have a memory unit, $S_2$ would have to apply an end-to-end universal compression to $x^n$. However, the side information provided by $y^\MM$ at $R$ about the source parameter can potentially result in a reduction in the amount of bits required to be transmitted in the $S_2$-$R$ link.
  Throughout the paper, we refer to this problem setup as Universal Compression of Distributed Identical Sources (UC-DIS).

In the study of coding strategies for UC-DIS, we compare the following cases for the compression of $x^n$ at $S_2$. Note that we assume that $y^m$ is already universally compressed at $S_1$ and transmitted and decoded at $R$.
\begin{itemize}
\item {{\NM}} (Universal compression), which only applies end-to-end  lossless universal compression to $x^n$  at $S_2$ without regard to $y^m$.
\item { {\CM}} (Universal compression with memorization at both the encoder and the decoder), which assumes that the encoder (at $S_2$) and the decoder (at $R$) have access to a common memory (i.e., sequence $y^\MM$), which is utilized in the  lossless compression of $x^n$ at $S_2$.
\item { {\WDCM}} (Distributed universal compression with memorization at the decoder), which assumes that decoder (at $R$) has memorized (i.e., has access to) $y^\MM$ while the encoder (at $S_2$) only knows the length $m$ of the side information but does not know the exact sequence $y^m$. The encoder then applies an almost lossless code to $x^n$ that is decoded at $R$ with permissible error probability $\PE$ using $y^m$.
\end{itemize}
Note that {\NM} does not benefit from the memorization and is the conventional scheme.
Further, {\CM} is introduced as the benchmark for the purpose of evaluating the performance of {\WDCM} and is not practically useful since it requires the sequence $y^m$ from $S_1$ to be available at the encoder of $S_2$.

Let $l_n(x^n)$ denote the strictly lossless length of  the codeword associated with the sequence $x^n$. Further, let $L_n^{}$  denote the space of  strictly lossless universal length functions on a sequence of length $n$.
Denote $R_n(l_n,\theta)$  as the expected redundancy of such strictly lossless codes on a sequence of length $n$ for the parameter vector $\theta$, i.e., $R_n(l_n,\theta) = \mb{E}l_n(X^n) -  H_n(\theta)$. Further, denote $\bar{R}_\NM(n)$ as the average minimax redundancy as given by
\begin{equation}
\bar{R}_\NM(n) \triangleq \min_{l_n \in L_n} \sup_{\theta \in \Theta^d} R_n(l_n,\theta).
\vspace{-.06in}
\end{equation}

In {\CM}, let $l_{n|\MM}$  be the strictly lossless universal length function with a memory sequence of length $m$.
Denote $L_{n|\MM}$  as the space of such strictly lossless universal length functions.
Let $R_n(l_{n|\MM},\theta)$  be the expected redundancy of encoding a sequence of length $n$ form the source $\mu_\theta$ using the length function $l_{n|\MM}$. Further, let $\bar{R}_\CM(n,m)$ denote the corresponding average minimax redundancy, i.e.,
\begin{equation}
\bar{R}_{\CM}(n,m) \triangleq \min_{l_{n|\MM} \in L_{n|\MM}} \sup_{\theta\in \Theta^d} R_n(l_{n|\MM},\theta).
\vspace{-.06in}
\end{equation}

In {\WDCM}, let $\hat{l}^{\PE}_{n|\MM}$ denote the almost lossless universal length function with a memorized sequence of length $\MM$ that is only available to the decoder, where the permissible error probability on decoding $x^n$ is $\PE$.
Further, denote $\hat{L}^{\PE}_{n|\MM}$  as the space of such universal length functions.
Denote $R_n(\hat{l}^{\PE}_{n|\MM},\theta)$  as the expected redundancy of encoding a sequence $x^n$ of length $n$  using the length function $\hat{l}^{\PE}_{n|\MM}$.
Denote $\bar{R}^{\PE}_\DCM(n,m)$ as the expected minimax redundancy as given by
\begin{equation}
\bar{R}^{\PE}_{\DCM}(n,m) \triangleq \min_{\hat{l}^{\PE}_{n|\MM} \in \hat{L}^{\PE}_{n|\MM}} \sup_{\theta\in \Theta^d} R_n(\hat{l}^{\PE}_{n|\MM},\theta).
\vspace{-.06in}
\end{equation}
Note that we denote $\bar{R}_\DCM(n,m) \triangleq \bar{R}^0_\DCM(n,m)$ as the expected minimax redundancy of \emph{strictly lossless} {\WDCM} coding strategy.

\section{Performance Evaluation of UC-DIS:\\Results on the Average Minimax Redundancy}
\label{sec:results}
In this section, we provide results on the average minimax redundancy of the different coding strategies introduced in the previous section for the UC-DIS problem.
Discussion on the implications of the results and the proof sketches are deferred to Sec.~\ref{sec:discussion} and Sec.~\ref{sec:analysis}, respectively.

In the case of strictly lossless {\NM}, Clarke and Barron derived the expected minimax redundancy $\bar{R}_\NM(n)$ for memoryless sources~\cite{Clarke_Barron}, which was later generalized by Atteson for Markov sources, as the following~\cite{atteson_markov}:
\begin{theorem}
The average minimax redundancy of strictly lossless {\NM} coding strategy is given by
\begin{equation}
\bar{R}_\NM(n) = \frac{d}{2}  \log\left( \frac{n}{2\pi e} \right) + \log \int
|\mc{I}_n(\theta)|^{\frac{1}{2}}d\theta + O\left(\frac{1}{n}\right),\nonumber
\label{eq:minimax}
\end{equation}
where $\mc{I}_n(\theta)$ is the Fisher information matrix.
\label{thm:NM}
\end{theorem}

In the case of strictly lossless {\CM} (i.e., when the two encoders can communicate), we obtain the average minimax redundancy in the following theorem.
\begin{theorem}
The average minimax redundancy of strictly lossless {\CM} coding strategy is given by
\begin{equation}
\bar{R}_{\CM}(n,m) = \frac{d}{2} \log\left( 1+ \frac{n}{\MM}\right)+ O\left(\frac{1}{\MM} + \frac{1}{n} \right).\nonumber
\end{equation}
\label{thm:CM}
\end{theorem}

In the next proposition, we confine ourselves to strictly lossless codes in the {\WDCM} strategy.
\begin{proposition}
The average minimax redundancy of strictly lossless {\WDCM} coding strategy is equal to that of {\NM} coding strategy. That is $\bar{R}_\DCM (n,m)= \bar{R}_\NM(n)$.
\label{thm:DCM}
\end{proposition}

Finally, in the case of almost lossless {\WDCM}, our main result is given in the following theorem.
\begin{theorem}
The average minimax redundancy of almost lossless {\WDCM} coding strategy is upper bounded by
\begin{equation}
\bar{R}^{\PE}_\WDCM(n,m)\hspace{-0.01in}\leq \hspace{-0.01in} \bar{R}_{\CM}(n,m)\hspace{-0.01in} +\hspace{-0.01in} \mc{F}(d,\PE) +\hspace{-0.02in} O\left( \hspace{-0.01in}\frac{1}{m}\hspace{-0.01in}+\hspace{-0.01in}\frac{1}{n}\right)\hspace{-0.02in},\nonumber
\end{equation}
where $\mc{F}(d,\PE)$ is the penalty term due to the encoders not communicating, which is given by
\begin{equation}
\mc{F}(d,\PE) =  \frac{d}{2} \log \left( 1  + \frac{2}{d\log e} \log \frac{1}{\PE} \right).
\end{equation}
\label{thm:WDCM}
\end{theorem}

\section{Discussion on the Results}
\label{sec:discussion}
In this section, we provide some discussion on the significance of the results for different UC-DIS coding strategies. Figures~\ref{fig:distributed_memoryless} and~\ref{fig:distributed_markov} demonstrate the redundancy rate for the three
coding strategies, namely, {\NM}, {\CM}, and {\WDCM} for memoryless sources and first-order Markov sources with alphabet
size $k=256$, respectively. In the case of {\NM}, Theorem~\ref{thm:NM} defines the achievable average minimax redundancy
for the compression of a sequence of length $n$ encoded without regard to the previously seen sequence $y^m$.

According to Theorem~\ref{thm:CM}, if the encoder and the decoder have access to a common memory $y^m$, i.e., {\CM} coding strategy, the average minimax redundancy could be much smaller than that of {\NM} depending on how large $m$ is. In particular, when $m \to \infty$ we have $\lim_{m \to \infty} \bar{R}_{\CM}(n,m) = 0$.\footnote{In this paper, we ignored the integer constraint on the length functions, which results in a negligible $O(1)$ redundancy analyzed in~\cite{Precise_Minimax_Redundancy, huffman_redundancy}.} This corresponds to the case where the parameter vector is known to both the encoder and the decoder, and thus, the redundancy is zero similar to a perfect Shannon code. Hence, the fundamental limits are those of known source parameters and universality no longer imposes a compression overhead. This is also demonstrated in Figs.~\ref{fig:distributed_memoryless} and~\ref{fig:distributed_markov}, where $m$ has been chosen to be sufficiently large.

Proposition~\ref{thm:DCM} demonstrates that if strictly lossless {\WDCM} coding strategy (i.e., $\PE = 0$) is to be used for the compression of $x^n$ from $S_2$, the memorization of $y^m$ from $S_1$ only at the decoder does not provide any compression benefit, assuming that the two encoders at $S_1$ and $S_2$ do not communicate. In other words, the best that $S_2$ can do is to simply apply a traditional universal compression on $x^n$.

Finally, according to Theorem~\ref{thm:WDCM}, unlike the asymptotic behavior of the Slepian-Wolf problem, the distributed nature in this problem incurs an extra redundancy on the compression. As can be seen in Fig.~\ref{fig:distributed_memoryless}, the overhead can be significant in the compression of memoryless sources. For example, when $n=512$B, $m=32$kB, and $\PE = 10^{-6}$, the redundancy rate is around $0.05$, as compared with the almost zero redundancy rate of {\CM}.
On the other hand, as demonstrated in Fig.~\ref{fig:distributed_markov}, when $d$ is relatively larger, for medium length sequences even with extremely small error probability, {\WDCM} performs fairly close to {\CM}. Further, {\WDCM} by far outperforms {\NM} in the compression of short to medium length sequences with reasonable permissible error probability, justifying usefulness of {\WDCM} in practice.
If $\log \frac{1}{\PE} \ll d$, the penalty term can be further simplified to be approximately equal to $\mc{F}(d,\PE) \approx \log \frac{1}{\PE}$ for the practical ranges of $\PE$.

\begin{figure}[tb]
\centering
\vspace{-0.05in}
\psfrag{ylabel}{\hspace{-0.2in}\footnotesize{Redundancy rate}}
\psfrag{xlabel}{\footnotesize{$n$}}
\psfrag{dataaaaaaaaaaaaaaaaaaaaaaa1}{\tiny{\NM}}
\psfrag{data2}{\tiny{{\WDCM} ($\PE = 10^{-40}$)}}
\psfrag{data3}{\tiny{{\WDCM} ($\PE = 10^{-6}$)}}
\psfrag{data4}{\tiny{\CM}}
\psfrag{bb1}{\tiny{$512$B}}
\psfrag{bb2}{\tiny{$4$kB}}
\psfrag{bb3}{\tiny{$32$kB}}
\psfrag{bb4}{\tiny{$256$kB}}
\epsfig{width=\figurewidth,file=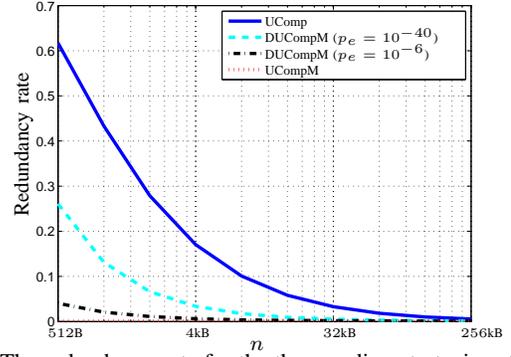}
\vspace{-0.15in}
\vspace{-5pt}
\caption{\small{The redundancy rate for the three coding strategies of interest for the UC-DIS problem. Memory size is $m=32$kB and the source is memoryless with alphabet size $k=256$.}}
\vspace{-.1in}
\label{fig:distributed_memoryless}
\end{figure}

\begin{figure}[tb]
\centering
\vspace{-0.05in}
\psfrag{ylabel}{\hspace{-0.2in}\footnotesize{Redundancy rate}}
\psfrag{xlabel}{\footnotesize{$n$}}
\psfrag{dataaaaaaaaaaaaaaaaaaaaaaa1}{\tiny{\NM}}
\psfrag{data2}{\tiny{{\WDCM} ($\PE = 10^{-40}$)}}
\psfrag{data3}{\tiny{{\WDCM} ($\PE = 10^{-6}$)}}
\psfrag{data4}{\tiny{\CM}}
\psfrag{bb1}{\tiny{$128$kB}}
\psfrag{bb2}{\tiny{$1$MB}}
\psfrag{bb3}{\tiny{$8$MB}}
\psfrag{bb4}{\tiny{$64$MB}}
\epsfig{width=\figurewidth,file=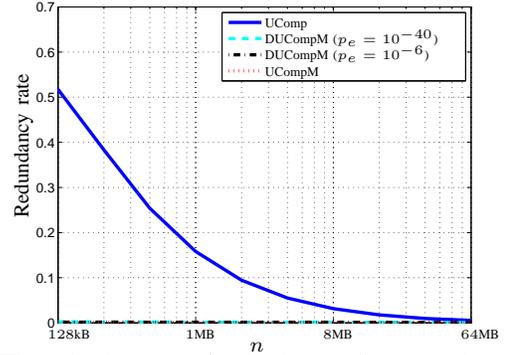}
\vspace{-0.15in}
\vspace{-5pt}
\caption{\small{The redundancy rate for the three coding strategies of interest for the UC-DIS problem. Memory size is $m=16$MB and the source is first-order Markov with alphabet size $k=256$.}}
\vspace{-.2in}
\label{fig:distributed_markov}
\end{figure}

\section{Technical Analysis}
\label{sec:analysis}
\subsection{Sketch of the Proof of Theorem~\ref{thm:CM}}
We prove that the RHS is both an upper bound and a lower bound for $\bar{R}_\CM(n,m)$. The upper bound is obtained using the KT-estimator~\cite{KT_estimator} along with a proper Shannon code~\cite{Shannon_paper} and the proof follows the analysis of the redundancy of the KT-estimator.
In the next lemma, we obtain the lower bound.
\begin{lemma}
The average minimax redundancy of {\CM} is lower-bounded by
\begin{equation}
\bar{R}_\CM(n,m) \geq \frac{d}{2} \log\left( 1+ \frac{n}{\MM}\right) + O\left( \frac{1}{\MM} + \frac{1}{n}\right).\nonumber
\vspace{-0.14in}\end{equation}
\label{lem:CTW_LB}
\end{lemma}
\begin{proof}
It can be shown that the minimax redundancy is equal to the capacity of the channel between the unknown parameter vector $\theta$ and the sequence $x^n$ given the sequence $y^\MM$ (cf. \cite{Merhav_Feder_IT} and the references therein). Thus,
\begin{eqnarray}
\bar{R}_\CM(n,m)
&\hspace{-5pt} = \hspace{-5pt}
&\sup_{\omega(\theta)} I(X^n;\theta|Y^\MM)\nonumber\\
&\hspace{-5pt} = \hspace{-5pt} & \sup_{\omega(\theta)} \{I(X^n, Y^\MM;\theta) - I(Y^{\MM};\theta)\}\nonumber\\
&\hspace{-5pt} \geq  \hspace{-5pt} & \{I(X^n, Y^\MM;\theta) - I(Y^{\MM};\theta)\}|_{\theta \propto \omega_\text{J}(\theta)}\nonumber\\
&\hspace{-5pt} = \hspace{-5pt} & \bar{R}_\NM({n+\MM}) - \bar{R}_\NM({\MM}),
\label{eq:max_I}
\end{eqnarray}
where $\omega_\text{J}(\theta) \triangleq \frac{|\mc{I}(\theta)|^{\frac{1}{2}}}{\int |\mc{I}(\beta)|^{\frac{1}{2}}d\beta}$ denotes the Jeffreys' prior, and  $\bar{R}_\NM(\cdot)$ is given in Theorem~\ref{thm:NM}. Further simplification of~(\ref{eq:max_I}) leads to the desired result in Lemma~\ref{lem:CTW_LB}.
\end{proof}

\subsection{Sketch of the Proof of Proposition~\ref{thm:DCM}}
Since the source is assumed to be from the family $\mc{P}^d$ of $d$-dimensional parametric sources, in particular, it is also an {\em ergodic} source. Thus, any pair $(x^n, y^m)$ occurs with non-zero probability and the support set of $(x^n,y^m)$ is equal to $\mc{A}^n \times \mc{A}^m$. Therefore, Proposition~\ref{thm:DCM} trivially follows from the known results on strictly lossless compression (cf.~\cite{Alon_source_coding} and the references therein).

\subsection{Sketch of the Proof of Theorem~\ref{thm:WDCM}}
We provide a constructive optimal coding strategy at the encoder and obtain its achievable average minimax redundancy, which provides with an upper bound on the average minimax redundancy of the almost lossless {\WDCM} coding strategy.

Let $\hat{\theta}(x^n)$ (or $\hat{\theta}(y^m)$) denote the Maximum Likelihood (ML) estimate for the unknown source parameter given that the sequence $x^n$ (or $y^m$) is observed, i.e., $\hat{\theta}(x^n) \triangleq \arg\max_\lambda \mu_\lambda(x^n)$. Further, let $\hat{\theta}_X \triangleq \hat{\theta}(x^n)$ and $\hat{\theta}_Y \triangleq \hat{\theta}(y^m)$.
As discussed earlier $\mu_\theta(x^n)$ is the probability distribution induced by the parameter vector $\theta$ on the sequence $x^n$. It is straightforward to derive the pmf of the ML-estimate $p(\hat{\theta}_X|\theta)$ from $\mu_\theta(x^n)$ by summing over all the sequences that correspond to the same ML-estimate. Note that $\hat{\theta}_X$ follows a discrete distribution only taking values on a finite set of $(n+1)^d$ points in the space $\Theta^d$.
For any $\lambda, \theta \in \Theta^d$, let $D_n(\mu_\lambda||\mu_\theta)$ be the KL-divergence, i.e.,
$
D_n(\mu_\lambda||\mu_\theta) \triangleq \mb{E} \log\left( \frac{\mu_\theta(X^n)}{\mu_\lambda(X^n)}\right)
$.
It can be shown that  expectations with respect to $p(\hat{\theta}_X|\theta)$ can be performed using a continuous RV $\tilde{\theta}_X$ (with  uniformly vanishing error) whose distribution conditioned on $\theta$ is given by
\begin{equation}
p(\tilde{\theta}_X|\theta) = |\mc{I}_n(\tilde{\theta}_X)|^{\frac{1}{2}} \left(\frac{n}{2\pi} \right)^{\frac{d}{2}} \exp(-D_n(\mu_{\tilde{\theta}_X}||\mu_\theta)),
\end{equation}
where $n$ has to be large enough so that Stirling's approximation can be applied.
Further, it is straightforward to show that this distribution can be approximated using a Gaussian distribution with mean $\theta$ and inverse covariance matrix $n \mc{I}_n(\theta)$.

Next, we will obtain an approximation for the distribution of $\hat{\theta}_X$ conditioned on $\hat{\theta}_Y$.
\vspace{-.04in}
\begin{lemma}
Let $\hat{\theta}_X$ and $\hat{\theta}_Y$ denote the ML-estimate parameter given observed sequences $x^n$ and $y^m$, respectively. Further, let $p(\tilde{\theta}_X | \hat{\theta}_Y)$ follow a Gaussian distribution with mean $\hat{\theta}_Y$ and inverse covariance matrix $\frac{nm}{n+m} \mc{I}_m(\hat{\theta}_Y)$. Then, all expectations with respect to $p(\hat{\theta}_X | \hat{\theta}_Y)$ can be performed using $p(\tilde{\theta}_X | \hat{\theta}_Y)$ with uniformly vanishing error.
\end{lemma}

Now, we are equipped to define $S_n(y^m, \PE)$ as the set with smallest Lebesgue volume such that
\begin{equation}
\int_{\tilde{\theta}_X \in \mc{S}_n(y^m,\PE)} p(\tilde{\theta}_X|\hat{\theta}_Y) d\tilde{\theta}_X \geq 1- \PE.
\end{equation}
The following lemma shows as to how $S_n(y^m, \PE)$ is determined.
\vspace{-.04in}
\begin{lemma}
Let $\hat{\theta}_Y$ denote the ML-estimate for the unknown parameter vector given sequence $y^m$ is observed. Then, $\mc{S}_n(y^m, \epsilon)$ is given by
\begin{equation}
\mc{S}_n(y^m, \PE)= \left\{\phi: r(\phi - \hat{\theta}_Y)' \mc{I}_m(\hat{\theta}_Y)(\phi - \hat{\theta}_Y) \leq  \delta_d(\PE) \right\},\nonumber
\end{equation}
where $r = \frac{nm}{n+m}$, $\delta_d(\PE)$ satisfies
$\Gamma\left(\frac{d}{2}, \delta_d(\PE)\right) = \PE \Gamma\left(\frac{d}{2}\right)$.\footnote{$\Gamma(s,x) \triangleq  \int_0^x{t^{s-1}e^{-t}dt}$ denotes the incomplete Gamma function.}
\end{lemma}
The next lemma determines the probability measure of the set $\mc{S}_n(y^m, \PE)$ under Jeffreys' prior.
\vspace{-.04in}
\begin{lemma}
Assume that the parameter vector $\theta$ follows Jeffreys' prior.
Then, the probability measure $P_\mc{S}(\PE)$ of the set $\mc{S}_n(y^m, \PE)$ is given by
\begin{equation}
P_\mc{S}(\PE) = \int_{\theta \in S_n(y^m,\PE)} \hspace{-.07in}\omega_\text{J}(\theta) d\theta= \frac{C_d}{\int |\mc{I}(\beta)|^\frac{1}{2}d\beta} \left(\frac{2\delta_d(\PE)}{r\log e}\right)^{\frac{d}{2}}\hspace{-.07in},\nonumber
\end{equation}
where $r = \frac{nm}{n+m}$ and $C_d = \frac{\Gamma\left(\frac{1}{2}\right)^d}{\Gamma\left( \frac{d}{2} + 1\right)}$.
\label{lemma:ellipsoid}
\end{lemma}

Next, consider the following coding scheme. Let the space be partitioned into ellipsoids of the form $S_{n}(y^m,\PE)$. Then, each sequence is encoded within its respective ellipsoid without regard to the rest of the parameter space.
The decoder chooses the decoding ellipsoid using the ML estimate $\hat{\theta}_Y$ and the permissible decoding error probability $\PE$ . The probability measure covered by each ellipsoid is $P_\mc{S}(\PE)$ is independent of $\hat{\theta}_Y$, and provides with $- \log P_\mc{S}(\PE)$ reduction in the redundancy. Further, simplification of $P_\mc{S}(\PE)$ and the fact that $\delta_d(\PE)\approx \frac{d}{2} \log e + \log \frac{1}{\PE}$ will lead to the desired result.

\section{Conclusion}
\label{sec:conclusion}
In this paper, we introduced and studied the problem of Universal Compression of Distributed Identical Sources (UC-DIS), which is a more favorable framework as compared to the {\SW} (SW) framework in several applications, such as the compression of data from mirrors of a data server.
In UC-DIS, the correlation among outputs of the sources is due to the finite-length universal compression constraint, departing from the nature of the correlation in the SW framework.
For UC-DIS, involving two identical sources, we introduced {\WDCM} coding strategy (compression using the side information at the decoder when the two encoders do not communicate) and obtained an upper bound on its average minimax redundancy.
We demonstrated that for finite-length sequences with reasonable permissible error probability, {\WDCM} coding strategy by far outperforms traditional universal compression, and hence, justifying the usefulness of {\WDCM} coding strategy in practice.



\ifCLASSOPTIONcaptionsoff
  \newpage
\fi



%

\bibliographystyle{IEEEtran}
\bibliography{compress}

\end{document}